\documentclass[11pt,letterpaper]{article}
\usepackage{amssymb}
\usepackage{amsmath}
\usepackage{amsthm}
\usepackage{fullpage}
\usepackage[ruled,vlined,linesnumbered]{algorithm2e}
\dontprintsemicolon

\DeclareMathOperator{\poly}{poly}
\newcommand{\eps}{\varepsilon}

\newtheorem{definition}{Definition}
\newtheorem{lemma}[definition]{Lemma}

\newtheorem{theorem}[definition]{Theorem}
\newtheorem{corollary}[definition]{Corollary}

\newcommand{\mathsub}[1]{_{\mbox{\scriptsize\rm #1}}}
\newcommand{\good}{\mathsub{good}}
\newcommand{\bad}{\mathsub{bad}}
\newcommand{\gneigh}{\mathsub{good-neighbor}}
\newcommand{\hcond}{\mathsub{high-conductance}}
\newcommand{\lcond}{\mathsub{low-conductance}}

\newcommand{\findneigh}{{\tt Find-Neighborhood}}
\begin{document}

\title{An Efficient Partitioning Oracle for Bounded-Treewidth Graphs}

\author{
Alan Edelman\\
MIT\\
{\tt edelman@mit.edu}
\and
Avinatan Hassidim\\
Google\\
{\tt avinatanh@mit.edu}
\and
Huy N. Nguyen\\
MIT\\
{\tt huy2n@mit.edu}
\and
Krzysztof Onak\thanks{Research supported by a Simons Postdoctoral Fellowship and NSF grants 0732334 and 0728645.}\\
CMU\\
{\tt konak@cs.cmu.edu}
}

\maketitle

\begin{abstract}
Partitioning oracles were introduced by Hassidim {\it et al.}\ (FOCS 2009) as a generic tool for constant-time algorithms. For any $\eps > 0$, a~partitioning oracle provides query access to a fixed partition of the input bounded-degree minor-free graph, in which every component has size $\poly(1/\eps)$, and the number of edges removed is at most $\eps n$, where $n$ is the number of vertices in the graph.

However, the oracle of Hassidim {\it et al.}\ makes an exponential number of queries to the input graph to answer every query about the partition. In this paper, we construct an efficient partitioning oracle for graphs with constant treewidth. The oracle makes only $O(\poly(1/\eps))$ queries to the input graph to answer each query about the partition.

Examples of bounded-treewidth graph classes include $k$-outerplanar graphs for fixed $k$, series-parallel graphs, cactus graphs, and pseudoforests. Our oracle yields $\poly(1/\eps)$-time property testing  algorithms for membership in these classes of graphs. Another application of the oracle is a $\poly(1/\eps)$-time algorithm that approximates the maximum matching size, the minimum vertex cover size, and the minimum dominating set size up to an additive $\eps n$ in graphs with bounded treewidth. Finally, the oracle can be used to test in $\poly(1/\eps)$ time whether the input bounded-treewidth graph is $k$-colorable or perfect. 
\end{abstract}

\section{Introduction}

Many NP-complete graph problems can be easily solved on graphs with bounded treewidth. For example, approximating vertex cover up to a multiplicative factor better than 1.36 is  known to be NP hard~\cite{DinurSafra05}.
In contrast, for graphs with bounded treewidth, one can in fact  find an optimal vertex cover in time that is linear in the size of the graph, but depends exponentially on treewidth \cite{bodlaender1988dynamic}. In this paper, we investigate yet another scenario in which bounded-treewidth graphs turn out to be much easier to deal with and allow for more efficient algorithms than general graphs.

\paragraph{Bounded Treewidth.}
The \emph{tree decomposition} and \emph{treewidth} were introduced by Robertson and Seymour \cite{robertson1984graph,robertson1986graph}, and later found many applications in the design of algorithms and machine learning (a nice, though outdated survey is \cite{bodlaender1994tourist}; see also \cite{atserias2008digraph,bodlaender2007combinatorial} for more recent applications). A \emph{tree decomposition} of a graph $G=(V,E)$ is a pair $(\mathcal{X},\mathcal{T})$, where $\mathcal{X} = (X_1, X_2, \ldots, X_m)$ is a family of subsets of $V$, and $\mathcal{T}$ is a tree (or forest) whose nodes are the subsets $X_i$, satisfying the following properties:
\begin{enumerate}
  \item Every $v \in V$ belongs to at least one $X_i$, i.e., $\bigcup_{i=1}^{m} X_i = V$.
  \item For every $(u,v) \in E$, there is an $X_i$ such that both $u$ and $v$ belong to $X_i$.
  \item For every vertex $v \in V$, the set of nodes in $\mathcal T$ associated with $v$ forms a connected subset of $\mathcal T$.
\end{enumerate}
The \emph{width of a tree decomposition} equals $\max_i|X_i| - 1$. The \emph{treewidth} of $G$ is defined as the minimum such width over all tree decompositions of $G$.

Graph families with bounded treewidth include $k$-outerplanar graphs, series-parallel graphs, cactus graphs, and pseudoforests.

\paragraph{The Bounded-Degree Model.} Before discussing our results, we describe the bounded-degree model introduced by Goldreich and Ron~\cite{GR_sparse} and used in this paper.
The degree of every vertex in the input graph is bounded by a constant $d$.
An algorithm can make two kinds of queries to access the input graph $G=(V,E)$.
First, for any vertex $v \in V$, it can obtain its degree $\deg(v)$ in constant time.
Second, for any vertex $v \in V$ and any $j$ such that $1 \le j \le \deg(v)$, the algorithm
can obtain the label of the $j$-th neighbor of $v$ in constant time.

The \emph{query complexity} of an algorithm $\mathcal{A}$ is the maximum number of queries made by $\mathcal{A}$ to the input graph. Also, the conventional \emph{time complexity} of $\mathcal{A}$ refers to maximum running time of $\mathcal{A}$. $\mathcal{A}$ is said to run in \emph{constant time} if its time complexity is independent of the number $n$ of vertices in the input graph.

\paragraph{Partitioning oracles.}
The main tool in (polynomial-time) approximation algorithms for minor-free graphs is the separator theorem \cite{LT_separator,LT80,AST90}, which shows a partition a graph into two components with a small number of edges connecting them. It is used to partition the original problem into independent subproblems, which are tractable. Stitching together the optimal solutions for each of the independent subproblems results in an approximate solution for the original problem.

In \cite{HKNO09}, this intuition was used to design constant-time algorithms for various problems. Fix a partition $P$ of the vertices in the input graph $G$ such that the following properties hold:
\begin{enumerate}
  \item Each connected component in the partition of the graph is small (say, has size $\poly(1/\epsilon)$).
  \item The number of edges connecting different connected components in the partition is less than $\eps|V|$.
\end{enumerate}
Suppose now that we are given query access to such a partition $P$, i.e., we have a procedure $\mathcal O$ which given a vertex $v$ returns the connected component $P(v)$ that contains $v$. We call such a procedure a \emph{partitioning oracle}. For a family of graphs $\mathcal F$, $\mathcal O$ is a partitioning oracle for $\mathcal F$ with parameter $\eps >0$ if it meets the following requirements:
\begin{itemize}
\item If $G \in \mathcal F$, then with probability $9/10$, the number of edges cut by the oracle is $\eps n$.

\item The oracle provides a partition of $G$, even if $G \not\in \mathcal F$.

\item The partition $P$, which $\mathcal O$ provides access to, is a function of only the input graph and random coin tosses of the oracle. In particular, $P$ cannot be a function of the queries to the oracle\footnote{This property allows algorithms to treat the partition $P$ as fixed, even if it is not explicitly computed for the entire graph until sufficiently many queries are performed.}.
\end{itemize}

We describe applications of partitioning oracles later, when we discuss the results that can be obtained using our oracle.

The main challenge here is to design efficient partitioning oracles that make few queries to the input graph and use little computation to answer every query about the partition.
\cite{HKNO09} shows how to design such an oracle for minor-free graphs (and also for some other hyperfinite graphs). However, their oracles have query complexity and running time of $2^{\poly (1/\epsilon)}$ (see~\cite{KOthesis} for a simple oracle with this property), and the main question left open by their paper is whether one can design a partitioning oracle that runs in time $\poly (1/\epsilon)$. In this paper we make partial progress by designing a partitioning oracle of complexity $\poly(1/\epsilon)$ for bounded-treewidth graphs.
A similar open question is posed in \cite{BSS08}, where they ask for a $\poly(1/\eps)$-time tester  for minor-closed properties. Constructing an efficient partitioning oracle for minor-free graphs would yield such a tester, but in general, such a tester need not be based on a~partitioning oracle.

\paragraph{Our Main Result.} The main result of the paper, an efficient partitioning oracle for bounded-treewidth graphs, is stated in the following theorem.

\begin{theorem}\label{theorem:partition-treewidth}
Let $G=(V,E)$ be a graph with maximum degree bounded by $d$.
Let $k$ be a positive integer.
There is an oracle $\mathcal O$ that given an $\eps \in (0,1/2)$,
and query access to $G$, provides query access to a function $f:V \to 2^V$ of the following properties (where $k = O\left(\frac{d^5 \cdot h^{O(h)} \cdot \log(d/\eps)}{\eps^3}\right)$):
\begin{enumerate}
\item For all $v \in V$, $v \in f(v)$.
\item For all $v \in V$, and all $w \in f(v)$, $f(v) = f(w)$.
\item If the treewidth of $G$ is bounded by $h$, then with probability $9/10$, $|\{(v,w) \in E: f(v) \ne f(w)\}| \le \eps|V|$.
\item For all $v \in V$, $|f(v)| \le k$.
\item For all $v$, the oracle 
makes $O(dk^{4h+7})$ queries to the input graph to answer a single query to the oracle, and the processing time is bounded by 
$\tilde O(k^{4h+O(1)} \cdot \log{Q})$, where $Q$ is the number of previous queries to the oracle.
\item The partition described by $f$ is a function of $G$ and random bits of the oracle, but does not depend on the queries to the oracle.
\end{enumerate}
\end{theorem}

\paragraph{Applications.} Partitioning oracles have numerous applications described in \cite{HKNO09}.
Let us describe some general applications of partitioning oracles, and the results yielded by our efficient partitioning oracle.
\begin{itemize}

\item {\bf Testing minor-closed properties:} In \emph{property testing} of graphs with maximum degree bounded by $d = O(1)$, the goal is to distinguish graphs that have a~specific property $P$ from those that need to have at least $\eps dn$ edges added and removed to obtain the property $P$, where $\eps > 0$ is a parameter.

Goldreich and Ron~\cite{GR_sparse} show that the property of being a tree can be tested in
$\tilde O(\eps^{-3})$ time. Benjamini, Schramm, and Shapira~\cite{BSS08} prove that any minor-closed property can be tested in $2^{2^{2^{\poly(1/\eps)}}}$ time. Hassidim {\it et al.}~\cite{HKNO09} introduce partitioning oracles and show how to use them to obtain a tester of running time $2^{\poly(1/\eps)}$ (see~\cite{KOthesis} for a~simplified full proof). Yoshida and Ito~\cite{YI10} show that outerplanarity and the property of being a cactus can be tested in $\poly(1/\eps)$ time.

Via the reduction from~\cite{HKNO09},
our new oracle yields a $\poly(1/\eps)$-time tester for any minor closed family of graphs that has bounded treewidth.
Sample minor-closed families of graphs with this property are $k$-outerplanar graphs, series-parallel graphs, and pseudoforests.
This also generalizes the result of Yoshida and Ito~\cite{YI10}, since outerplanar graphs have treewidth at most 2. 

\item {\bf Constant-time approximation algorithms:}
Our oracle can also be used to obtain a $\poly(1/\eps)$-time additive $\eps n$-approximation
algorithm for the size of the maximum matching, minimum vertex cover, and minimum dominating set
in (even unbounded degree) graphs with constant treewidth. See~\cite{HKNO09} for a general reduction.
An important fact here is that for bounded-treewidth graphs, there are linear-time algorithms for computing the exact solutions to these problems~\cite{ArnborgP89}.
This result adds to a long line of research on this kind of approximation algorithm
\cite{PR07,MR09,NO08,YYI09,HKNO09,Elek10,NS11}.

\item {\bf Testing properties of bounded-degree
bounded-treewidth graphs:}
Czumaj, Shapira, and Sohler~\cite{CSS} show that any hereditary property of bounded-degree bounded-treewidth graphs can be tested in constant time. This result is generalized by Newman and Sohler~\cite{NS11}, who show that in fact any property of such graphs can be tested in constant time. Unfortunately, these two papers do not yield very efficient algorithms in general. Using our oracle and another reduction from \cite{HKNO09}, one can show that there are $\poly(1/\eps)$-time algorithms for testing $k$-colorability and graph perfectness for these graphs. As before, one has to use efficient polynomial-time algorithms for these problems \cite{ArnborgP89,ChudnovskyCLSV05} to solve the exact decision problems for sampled components in the reduction from~\cite{HKNO09}.

\end{itemize}

\subsection{Overview of Our Techniques}

Let us briefly describe the main ideas behind the proof of our main result.
Let $G$ be a bounded-degree graph with treewidth $h$. We say that a vertex in $G$ has a~``good neighborhood'' if a small set $S$ of vertices including $v$ can be disconnected from the graph by deleting at most $O(h)$ other vertices. Moreover, $O(h)/|S|$ is small. 

First we show that most vertices have a good neighborhood. This follows by taking a tree decomposition of $G$, and showing a method that constructs a~partition in which most vertices end up in connected components that can play the role of $S$ for them in the original graph.

Then using the fact that a good neighborhood of small size $t$ has a small border, i.e., of size $O(h)$, we show a procedure for enumerating all good neighborhoods for a given vertex. The procedure runs in $\poly(dt)^{O(h)}$ time, where $t$ is a bound on the size of the neighborhood.  In particular, it can be used to check whether a given vertex has a good neighborhood and find it, if it exists.

Finally, we show a global partitioning algorithm that is likely to compute the desired partition of the graph. In this algorithm, each vertex $v$ computes an arbitrary good neighborhood $S_v$ containing it. If such a neighborhood does not exist, we set $S_v := \{v\}$ instead. Then we consider all veritices in $V$ in random order. In its turn, $v$ removes all the remaining vertices in $S_v$ from the graph. The set of vertices removed by $v$ constitutes one (or a constant number) of the connected components in the partition of the input graph. This algorithm can easily be simulated locally and is the basis of our partitioning oracle. 

\paragraph{Note:} An anonymous reviewer suggested using known results on tree-partition-width \cite{DingOporowski1995,Wood2009} to simplify our proofs. Those results can be used to give a simpler proof of a slightly worse decomposition than that in Lemma~\ref{lem:treewidthlocalcut}.
Namely, one can immediately combine them with Lemma~\ref{lem:treecut2} to obtain a partition of a bounded-treewidth graph in which almost most vertices belong to components that have neighborhoods of size at most $O(dh)$ as opposed to $O(h)$ in Lemma~\ref{lem:treewidthlocalcut}.
For constant $d$, the worse decomposition still results in the oracle's query complexity of $(1/\eps)^{\poly(h)}$. Unfortunately, in some application, for instance in the approximation of the minimum vertex cover size in graphs of arbitrary degree, $d = O(1/\eps)$ and the weaker partition eventually results in an algorithm that runs in $2^{\poly(1/\eps)}$ time. Our construction gives a $\poly(1/\eps)$-time algorithm.
\section{Definitions}

Let $G = (V,E)$ be a graph and $S$ be a subset of $V$. We write $N(S)$ to denote the set of vertices that are not in $S$, but are adjacent to at least one vertex in $S$. We write $\eta(S)$ to denote the \emph{cut-size} of $S$, which is defined as the size of $N(S)$, $\eta(S) = |N(S)|$. We write $\phi(S)$ to denote the \emph{vertex conductance} of $S$, which is defined as  $\phi(S) = \frac{\eta(S)}{|S|}$.

\begin{definition}
Let $G = (V,E)$ be a graph. We say that $S \subseteq V$ is a \emph{neighborhood} of $v$ in $G$ if $v \in S$ and the subgraph induced by $S$ is connected. Given $k, c \geq 1$ and $\delta \in (0,1)$, we say that $S$ is a $(k,\delta,c)$\emph{-isolated neighborhood} of $v \in V$ if $S$ is neighborhood of $v$ in $G$, $|S| \leq k$, $\eta(S) \leq c$ and $\phi(S) \leq \delta$.
\end{definition}

\begin{definition}
Let $G = (V,E)$ be a graph and let $A$ be a family of sets of vertices in $G$. A subfamily $B \subseteq A$ is a \emph{cover} of $A$ if for every set $T \in A$, $ T \subseteq \bigcup_{S \in B}{S}$.
\end{definition}

\section{Local Isolated Neighborhoods in Bounded-Treewidth Graphs}

The following lemma is at the heart of our proof. It shows that given a bounded-treewidth and bounded-degree graph, we can find an isolated neighborhood of $v$ for almost every vertex $v$ in the graph.

\newcommand{\const}{28860}
\begin{lemma}\label{lem:treewidthlocalcut}
Let $G=(V,E)$ be a graph with treewidth bounded by $h$ and maximum degree bounded by $d$.
For all $\eps, \delta \in (0,1/2)$, there exists a function $g:V \to 2^V$ with the following properties:
\begin{enumerate}
\item For all $v \in V$, $v \in g(v)$.
\item For all $v \in V$, $|g(v)| \leq k$, where $k = \frac{\const\, d^3  (h+1)^5 }{\delta \eps^2}$.
\item For all $v \in V$, $g(v)$ is connected.
\item Let $\mathcal B$ be the subset of $V$ consisting of $v$ such that
$g(v)$ is a $\left(k,\delta,2(h+1)\right)$-isolated neighborhood of $v$ in $G$. The size of $\mathcal{B}$ is at least $(1-\eps/20)|V|$.
\end{enumerate}
\end{lemma}

Before showing the proof for Lemma~\ref{lem:treewidthlocalcut}, we need a few results on forest partitioning and bounded-treewidth graphs.

\subsection{A Strong Partition for Forests}
The following lemma is a strong version of the Lemma~\ref{lem:treewidthlocalcut} for trees and forests (treewidth is at most $2$). We show that for any forest $T$, there is a partition of $T$ such that the cut-sizes of most of the parts in the partition are smaller or equal to $2$.

\begin{lemma} \label{lem:treecut2}
Let $T=(V_T,E_T)$ be a forest with maximum degree bounded by $d \geq 2$. Let $\eps,\delta \in (0,1/2)$, and let $k = \frac{481 d^2}{\delta \eps}$. There exists a partition $f:V \to 2^V$ with the following properties.
\begin{enumerate}
\item For all $v \in V$, $v \in f(v)$.
\item For all $v \in V$, and all $w \in f(v)$, $f(w) = f(v)$.
\item For all $v \in V$, $|f(w)| \le k$.
\item For all $v \in V$, the subgraph of $T$ induced by $f(v)$ is connected.
\item Let $\mathcal C$ be the subset of $V_T$ consisting of $w$ such that 
$f(w)$ is a $(k,\delta,2)$-isolated neighborhood of $w$ in $T$. The size of $\mathcal{C}$ is at least $(1-\eps/60)|V|$.
\end{enumerate}
\end{lemma}

\begin{proof}
\begin{algorithm}[ht]
\caption{\tt Stronger-Tree-Partitioning}
\label{alg:stronger-tree-partition}
$k' := \frac{480 d}{\delta \eps}$ \\
$G := T$\\

\tcc*[h]{Phase 1: Contract leaves of weight less than $k'$}\\
\lForAll{\rm vertex $v$}{$s[v] := \{v\}$\\}
\While {\rm there exists a vertex $v$ of degree 1 such that $|s[v]| < k'$} {
Let $u$ be the neighbor of $v$.\\
$s[u] := s[u] \cup s[v]$\\
Remove $v$ from G.
}

\tcc*[h]{Phase 2: Remove vertex with degree greater than 2}\\
\ForEach {vertex $v$ of degree greater than 2}  {
\ForEach{$u \in s[v]$ such that $(u,v) \in E_T$} { 
\lForEach{$w \in s[u]$}{$f(w) := s[u]$}
}
$f(v) := \{ v \}$\\
Remove $v$ from G.
}

\tcc*[h]{Phase 3: Partitioning paths}\\
\While {there exists a vertex $v$ in $G$ with degree 1} {
\If {$|s[v]| \geq 2/\delta$} {
\lForAll {$w \in s[v]$} {$f[w] := s[v]$} \\
Remove $v$ from $G$.
}
\Else {
Let $u$ be the only neighbor of $v$ in $G$. \\
$s[u] := s[u] \cup s[v]$ \\
Remove $v$ from $G$
}
}

\tcc*[h]{Phase 4: Partitioning isolated vertices}\\
\While {there exists an isolated vertex $v$ } {
\lForAll{$w \in s[v]$}{$f[w] := s[v]$ }\;
Remove $v$ from $G$.
}
\end{algorithm}
Consider the partitioning method described in Algorithm~\ref{alg:stronger-tree-partition}.
We show that it produces the desired partition.
The method consists of four parts:
\begin{enumerate}
	\item {\bf Steps 3--7:} Small branches are shrunk. For each vertex $v \in V_G$, $s[v]$ is the set of vertices in $V_T$ that are contracted to $v$.
	\item {\bf Steps 8--12:} In this part, the algorithm partitions the set of vertices contracted to vertices degree higher than $2$. For each vertex of degree higher than $2$, all the vertices in each branch that is shrunk to that vertex are turned into a single component in the partition and the vertex is turned into a component of size $1$. Once all the vertices with degree higher than $2$ are removed, the remaining graph contains only disjoints paths and isolated vertices.
	\item {\bf Steps 13--20:} In this part, the algorithm partitions the paths in the remainining  graph. It repeatedly looks for end vertices of the paths. For each end vertex, if its weight is less $2/\delta$, the vertex is contracted to its only neighbor. Otherwise, if the vertex's weight is at least $2/\delta$, all the vertices that are contracted to that vertex are turned into a single component in the partition.
	\item {\bf Steps 21--23:} Finally, for each remaining isolated vertex, the set of vertices that are contracted to that vertex are turned into a single component in the partition.
\end{enumerate}

Clearly, from the construction of the function $f$, Claims 1, 2 and 4 of the Lemma \ref{lem:treecut2} hold.

Let us bound the size of each component in the partition constructed by Algorithm~\ref{alg:stronger-tree-partition}. First, observe that for each component in the partition, all the vertices in that component are contracted to the same vertex. Also, observe that the weight of a vertex in $T$ is a most $d \cdot \max \{k', 2/\delta\} + 1 \le \frac{480d^2}{\delta\eps} + 1 \le k$. Therefore, the size of each component in the partition is also bounded by $k$, as required by Claim~3.

Finally, we show that $|\mathcal C| \ge (1-\eps/60)|V|$:
\begin{itemize}

	\item Let us bound the number of vertices in $T$ that belong to components with cut-size greater than $2$. Observe that for any $v \in V_T$, the cut-size of $f(v)$ is only greater than $2$ if $v$ is also a vertex of degree higher than $2$ in Phase 2 of Algorithm~\ref{alg:stronger-tree-partition}. Since after the first part of Algorithm~\ref{alg:stronger-tree-partition}, the weight of each leaf vertex in $G$ is at least $k'$, the number of leaves in $G$ at Step 8 is at most $\frac{1}{k'}|V|$. This also implies that the number of vertices with degree higher than $2$ in $G$ at Step 8 is bounded by $\frac{1}{k'}|V| \leq \frac{\eps}{480}|V|$.

	\item We now bound the number of vertices in $T$ that belong to components of conductance greater than $\delta$. Given $v \in V_T$, observe that $\phi(f(v))$ is only greater than $\delta$ in one of the following cases:
\begin{itemize}
	\item[-] $v$ is a vertex of degree higher than $2$ in Step 8 of Algorithm~\ref{alg:stronger-tree-partition}. There are at most $\frac{1}{k'}|V|  \leq \frac{\eps}{480}|V|$ vertices of this type.
	\item[-] $v$ is contracted to a vertex $u$, and $u$ is a vertex of degree higher than $2$ in Step 8. In this case, $\phi(f(v)) > \delta$ if and only if the size of the branch that was shrunk to $u$ and contains $v$, is less than $1/\delta$. For each high degree vertex $u$, there are at most $d$ such small branches. Therefore, the total number of vertices in $T$ of this type is at most $\frac{d}{\delta k'}|V| \leq \frac{\eps}{480}|V|$.
	\item[-] $v$ is contracted to a vertex $u$ such that $u$ is the last vertex that remains after partitioning a path in $G$ in Step 13 and the weight of $u$ is smaller than $2/\delta$. Observe that the number of paths in $G$ in Step 13 is bounded by twice the number of leaves in the forest in Step~8. Therefore, the number of such vertices $v$ is bounded by $\frac{2}{k'}|V| \cdot \frac{2}{\delta} \le \frac{\eps}{120}|V|$.
\end{itemize}
  Summarizing, the number of vertices in $T$ belonging to components of conductance greater than $\delta$ is at most $\frac{6}{480}\eps|V|$.
\end{itemize}
By the union bound, there are at least $(1-\eps/60)|V|$ vertices in $T$ that belong to components of conductance at most $\delta$ and of cut-size at most $2$. Combining this with the Claim 3 of the lemma, we have $|\mathcal C| \geq (1-\eps/60)|V|$, as stated in Claim 5.
\end{proof}

\subsection{Some Properties of the Tree Decomposition}

We first introduce a few helpful definitions.

\begin{definition}
Let $(\mathcal{X},\mathcal{T})$ be a tree decomposition of $G = (V,E)$, where $\mathcal{X} = (X_1, X_2, \cdots, X_m)$ is a family of subsets of $V$ and $\mathcal{T}$ is a forest whose nodes are the subsets $X_i$.
\begin{enumerate}
  \item We say $(\mathcal{X},\mathcal{T})$ is \emph{edge-overlapping} if for any $X_i, X_j \in \mathcal{X}$ such that $(X_i,X_j) \in \mathcal{T}$, $X_i \cap X_j \neq \emptyset$.
  \item We say $(\mathcal{X},\mathcal{T})$ is \emph{minimal} if for all $X_i \in \mathcal{X}$ and all $x \in X_i$, removing $x$ from $X_i$ makes $(\mathcal{X},\mathcal{T})$ no longer a tree decomposition of $G$.
  \item We say $(\mathcal{X},\mathcal{T})$ is \emph{non-repeated} if no $X_i = \emptyset$ and
  there is no $(X_i, X_j) \in \mathcal{T}$ such that $X_i \subseteq X_j$.

\item Let $\mathcal T'$ be a subforest of $\mathcal T$. We write $\mathcal X|_{\mathcal{T}'}$ to denote the set of nodes in $\mathcal T'$.

\item Let $\mathcal S$ be a subset of $\mathcal X$. We write $V|_{\mathcal S}$ to denote $\bigcup_{Z \in \mathcal S}{Z}$. 

\item Let $\mathcal T'$ be a subforest of $\mathcal T$, we write $V|_{\mathcal T'}$ to denote $V|_{\mathcal X|_{\mathcal T'}}$ (= $\bigcup_{Z \in \mathcal X|_{\mathcal{T}'}}{Z}$).

\item Let $v$ be a vertex in $V$, we say that $v$ \emph{appears} in a subset $\mathcal S \subseteq \mathcal X$ if $v \in V|_{\mathcal S}$, and $v$ \emph{appears} in a subforest $\mathcal T'$ of $\mathcal T$ if $v \in V|_{\mathcal T'}$.

\item Let $e = (u,v) \in E$ be an edge in $G$, we say that a subset $\mathcal S \subseteq \mathcal X$ \emph{witnesses} $e$ if there exists a node $Z \in \mathcal S$ such that $u,v \in Z$. Similarly, we say that a subforest $\mathcal T'$ of $\mathcal T$ \emph{witnesses} $e$ if there exists a node $Z \in \mathcal X|_{T'}$ such that $Z$ contains both $u$ and $v$.

\end{enumerate}
\end{definition}

The next three lemmas provide simple properties of tree decompositions.

\begin{lemma}\label{lem:canonical-tree-decomposition}
Let $G = (V,E)$ be a graph with treewidth $h$. There is a tree decomposition of $G$ of width $h$ that is edge-overlapping, minimal, and non-repeated.
\end{lemma}

\begin{proof}
Let $(\mathcal{X},\mathcal{T})$ be a tree decomposition of $G$ of width $h$. 
We repeat the following operations in arbitrary order as long as we can.
\begin{enumerate}
\item If there exists an empty node $X_i$, remove it.
  \item If there exists a pair of nodes $X_i, X_j \in \mathcal{X}$ such that $(X_i,X_j) \in \mathcal{T}$ and $X_j \subseteq X_i$, remove $X_j$ and connect all neighbors of $X_j$ to $X_i$.
	\item If there exists a pair of nodes $X_i, X_j \in \mathcal{X}$ such that $(X_i,X_j) \in \mathcal{T}$ and $X_i \cap X_j = \emptyset$, remove the edge $(X_i,X_j)$ from $\mathcal{T}$.
	\item If there exists a node $X_i \in \mathcal X$ and a vertex $x_i \in X_i$ such that $(\mathcal{X},\mathcal{T})$ is still a tree decomposition of $G$ after removing $x_i$ from $X_i$, then remove $x_i$ from $X_i$.
\end{enumerate}
The process has to stop because in every iteration, $|\mathcal X|+|\mathcal T|+\sum_i|X_i|$ decreases. When the process stops, the tree decomposition is
by definition edge-overlapping, minimal, and non-repeated.
\end{proof}

\begin{lemma} \label{lem:treewidth-size}
Let $G = (V,E)$ be a graph and let $(\mathcal{X},\mathcal{T})$ be a non-repeated tree decomposition of $G$ of width $h$. Let $\mathcal T'$ be a subtree of $\mathcal{T}$.
$$\frac{\Bigl|V|_{\mathcal{T}'}\Bigr|}{h+1} \leq \Bigl|\mathcal X|_{\mathcal{T}'}\Bigr| \leq \Bigl|V|_{\mathcal{T}'}\Bigr|.$$
\end{lemma}

\begin{proof}
The first inequality is straightforward. Since each vertex in $V|_{\mathcal{T}'}$ appears at least once in $\mathcal X|_{\mathcal{T}'}$ and each set in $\mathcal X|_{\mathcal{T}'}$ contains at most $h+1$ vertices, $\Bigl|\mathcal X|_{\mathcal{T}'}\Bigr| \geq \frac{\bigl|V|_{\mathcal{T}'}\bigr|}{h+1}$.

The second inequality can be proved by induction on the size of $\mathcal{X}(\mathcal{T}')$.
If $\mathcal{T}'$ is an isolated node $Y$ in $\mathcal T$, then $\Bigl|\mathcal X|_{\mathcal{T}'}\Bigr| = 1 \leq \Bigl|V|_{\mathcal{T}'}\Bigr| = |Y|$, where the last inequality holds because no $Y \in \mathcal X$ is empty.

For the inductive case, consider a subtree $\mathcal{T}'$ of $\mathcal T$ of size $k > 1$.
Let $Y \in \mathcal X|_{\mathcal{T}'}$ be a leaf vertex in $\mathcal{T}'$. By the induction hypothesis,
\begin{align}
\Bigl|X|_{\mathcal T'}\Bigr| - 1 = \Bigl|\mathcal X|_{\mathcal{T}'} \setminus \{Y\}\Bigr| \leq \left|\bigcup_{Z \in \mathcal X|_{\mathcal{T}'} \setminus \{Y\}}{Z}\right|. \label{tag1}
\end{align}
Let $Y'$ be $Y$'s neighbor in $\mathcal{T}'$. Since $(\mathcal{X},\mathcal{T})$ is non-repeated, there is a vertex $v \in V$ such that $v \in Y \setminus Y'$. 
By the definition of the tree decomposition, the set $\{ Z \in \mathcal X : v \in Z \}$ is a connected component of $\mathcal{T}$. Therefore, for every $Z \in \mathcal X|_{\mathcal{T}'} \setminus \{ Y \}$, $v \notin Z$. 
This implies that
\begin{align}
\left|\bigcup_{Z \in \mathcal X|_{\mathcal{T}'} \setminus \{ Y \} }{Z}\right| \leq \Bigl|V|_{\mathcal{T}'} \setminus \{v\}\Bigr| = \Bigl|V|_{\mathcal{T}'}\Bigr| - 1 \label{tag2}.
\end{align}
From (\ref{tag1}) and (\ref{tag2}), we have $\Bigl|\mathcal X|_{\mathcal{T}'}\Bigr| \leq \Bigl|V|_{\mathcal{T}'}\Bigr|$, which finishes the inductive proof.
\end{proof}

\begin{lemma} \label{lem:treewidth-degree}
Let $G = (V,E)$ be a graph with maximum degree bounded by $d$ and let $(\mathcal{X},\mathcal{T})$ be its edge-overlapping and minimal tree decomposition of width $h$. The maximum degree of a node in $\mathcal{T}$ is bounded by $d(h + 1)$.
\end{lemma}

\begin{proof}
Let $Y \in \mathcal X$ be a node in $\mathcal{T}$. Consider a neighbor $Y'$ of $Y$. Since $(\mathcal{X},\mathcal{T})$ is edge-overlapping, there must exist some $x \in V$ such that $x \in Y \cap Y'$. In addition, consider the subtree $\mathcal T'$ containing $Y'$ that is created by removing $Y$. Since $(\mathcal{X},\mathcal{T})$ is minimal, $\mathcal T'$ must witness some edge $(x,y) \in E$, for $y \in V$, that is not witnessed by any other subtree created by removing $Y$. Therefore, we can assign to each new subtree created by removing $Y$ a unique edge incident to a vertex in $Y$. The degree of $Y$ in $\mathcal T$ equals exactly the number of the new subtrees that the connected component of $Y$ breaks into after removing $Y$. Since there are at most $h+1$ vertices in $Y$, each of degree at most $d$, $Y$ has at most $d(h+1)$ neighbors in $\mathcal T$.
\end{proof}

\subsection{Proof of Lemma~\ref{lem:treewidthlocalcut}}

We are finally ready to prove the main lemma of this section.

\begin{proof}[Proof of Lemma \ref{lem:treewidthlocalcut}]
By Lemma \ref{lem:canonical-tree-decomposition}, there exists an edge-overlapping, minimal and non-repeated tree decomposition $(\mathcal{X},\mathcal{T})$ of $G$ of width $h$. By Lemma \ref{lem:treewidth-degree}, $\mathcal{T}$ is a forest with maximum degree bounded by $d \cdot (h+1)$. By Lemma \ref{lem:treecut2}, with the lemma's $\eps$ set to $\eps/(h+1)$, $d$ set to $d(h+1)$ and $\delta$ set to $\frac{\delta \eps}{60d(h+1)}$, there exists a partition $f : \mathcal X \to 2^{\mathcal X}$ such that:
\begin{enumerate}
	\item For all $X \in \mathcal{X}$, $X \in f(X)$.
	\item For all $X \in \mathcal{X}$, and all $Y \in f(X)$, $f(Y) = f(X)$.
	\item For all $X \in \mathcal{X}$, $|f(X)| \le \frac{\const\, d^3  (h+1)^4 }{\delta \eps^2}$.
	\item For all $X \in \mathcal{X}$, the subgraph of $\mathcal T$ inducted by $f(X)$ is connected.
	\item The size of $\mathcal C$, the subset of vertices in $\mathcal X$ such that for every $X \in \mathcal C$, $f(X)$ is a~$\left(\frac{\const\, d^3  (h+1)^4 }{\delta \eps^2},\frac{\delta\eps}{60d(h+1)},2\right)$-isolated neighborhood of $X$ in $T$, is at least $(1-\eps/60(h+1))|\mathcal X|$.
\end{enumerate}
Let $\mathcal P$ be the set of all components in the partition $f$. To simplify the notation, we say that a component $P$ in $\mathcal P$ is \emph{good} if $P$ is a $\left(\frac{\const d^3  (h+1)^4 }{\delta \eps^2},\frac{\delta}{60d(h+1)},2\right)$-isolated neighborhood. Otherwise, $P$ is \emph{bad}. Let $\mathcal P\good$ be the set of good components in $\mathcal P$.
Similarly, we define $\mathcal P\bad = \mathcal P \setminus \mathcal P\good$ as the set of bad components in $\mathcal P$. Let $A\bad = \bigcup_{P \in \mathcal P\bad}{V|_{P}}$, i.e., $A\bad$ is the set of vertices in $V$ that appear in at least one bad component of the partition. Let $A\gneigh = \bigcup_{P \in \mathcal P\good}{V|_{N(P)}}$ is the set of vertices in $V$ that appear in the neighborhood of at least one good component\footnote{Recall that $N(P)$ denotes the set of nodes in $\mathcal T$ that are not in $P$, but are adjacent to at least one vertex in $P$.}. Finally, let $A\good = V \setminus (A\bad \cup A\gneigh)$ be the set of vertices in $V$ such that for every $v \in A\good$, $v$ only appears in exactly one component in $\mathcal P$ and that component is good.

Let us construct the function $g$ as follows. For each $v \in V$, if $v \notin A\good$, set $g(v)$ to $\{v\}$. Otherwise, if $v \in A\good$, set $g(v)$ to the connected component that contains $v$ in the subgraph of $G$ induced by $A\good$.

It is clear from the construction that Claims 1 and 3 of Lemma \ref{lem:treewidthlocalcut} hold for $g$. Let us bound the size of $g(v)$. If $v \notin A\good$, then it is clear from the definition of $g$ that $|g(v)| = 1$. Otherwise, if $v \in A\good$, let $P$ be the only component in which
$v$ appears. It is clear from the construction of $g$ that $V|_{N(P)} \cap A\good = \emptyset$. Therefore, $g(v) \subseteq V|_{P}$. 
Thus, $|g(v)| \leq \Bigl|V|_{P}\Bigr| \leq (h+1)|P| \leq \frac{\const\, d^3  (h+1)^5 }{\delta \eps^2}$. In both the cases, the size of $g(v)$ is bounded by $\frac{\const\, d^3  (h+1)^5 }{\delta \eps^2}$, as stated in Claim 2 of the lemma.

Finally, we now show that $|\mathcal B| \geq (1-\frac{\eps}{20})|V|$. Consider a vertex $v \in A\good$. Let $P$ be the only component in $\mathcal P$ in which $v$ appears. Observe that every edge going out of $g(v)$ must end up in $V|_{N(P)}$. Therefore,
$$\eta(g(v)) \leq V|_{N(P)} \leq (h+1) |N(P)| \leq 2(h+1). $$

Let $A\hcond$ be the subset of $A\good$ consisting of $v$ such that  $\phi(g(v)) > \delta$. Also, let $A\lcond = A\good \setminus A\hcond$. For all $v \in A\lcond$, $g(v)$ is a $(k,\delta,2(h+1))$-isolated neighborhood of $v$ in $G$, that is, $A\lcond \subseteq \mathcal B$. It therefore suffices to lower-bound the size of $A\lcond$, which we do next.

Observe first that for each bad component $P \in \mathcal P$, all nodes in $P$ belong to $\mathcal{X} \setminus \mathcal{C}$,
	$$ |A\bad| \leq \left|\bigcup_{Z \in \mathcal X \setminus \mathcal C}{Z}\right| \leq (h+1)\cdot |\mathcal X \setminus \mathcal C| \leq \frac{\eps}{60} |\mathcal X| \leq \frac{\eps}{60}|V|,$$
	where the last inequality follows from Lemma~\ref{lem:treewidth-size}.
Observe next that for each good component $P \in \mathcal P$, $|N(P)| \leq \frac{\delta\eps}{60d(h+1)} |P|$. Thus, $$\Bigl|V|_{N(P)}\Bigr| \leq (h+1) \cdot  |N(P)| \leq \frac{\delta\eps}{60d} |P|.$$
	This implies that $|A\gneigh| \leq \frac{\delta\eps}{60d} |V| \leq \frac{\eps}{60} |V|$.
Note now that since every edge leaving $g(v)$ goes to $A\gneigh$, the total number of edges leaving $g(v)$ for all $v \in A\good$ is at most $d \cdot |A\gneigh|$. Thus, $$|A\hcond| \le \frac{d}{\delta} \cdot \left|A\gneigh\right| \leq \frac{\eps}{60}\cdot|V|.$$
Finally, we obtain
\begin{eqnarray*}
|\mathcal B| &\geq&  |A\lcond| = |A\good| - |A\hcond| \\
&\ge& |V| - |A\bad| - |A\gneigh| - |A\hcond| \geq \left(1-\frac{\eps}{20}\right)|V|.
\end{eqnarray*}
\end{proof}

\section{Isolated Neighborhoods}\label{sec:find-isolated-neighborhood}

In this section we show how to discover isolated neighborhoods efficiently. We also prove an upper bound on the number of incomparable isolated neighborhoods covering a specific vertex.

\subsection{Finding an Isolated Neighborhood of a Vertex}

The following lemma states that isolated neighborhoods with small cut-size in bounded-degree graphs can be found efficiently. To this end, we use the procedure \findneigh{} described as Algorithm~\ref{proc:local-bounded-cut-neighborhood}.

\begin{lemma}\label{lem:find-isolated-neighborhood}
Let $G = (V,E)$ be a graph with maximum degree bounded by $d$. Given a vertex $v \in V$, integers $k,c \geq 1$ and $\delta \in (0,1)$, procedure \findneigh{} finds a $(k,\delta,c)$-isolated neighborhood of $v$ in $G$, provided it exists.
If no such isolated neighborhood exists, the algorithm returns $\{v\}$. The algorithm runs in $\poly(dk) \cdot k^{c}$ time and makes $O(dk^{c+1})$ queries to the graph.
\end{lemma}

\begin{proof}
\begin{algorithm}[t]
\caption{Procedure \findneigh($v$,$k$,$\delta$,$c$)}
\label{proc:local-bounded-cut-neighborhood}
Run BFS from $v$ until it stops or exactly $k$ vertices are visited \\
Let $S$ be the set of vertices reached by the BFS \\
\If {$S$ is a $(k,\delta,c)$-isolated neighborhood in the original graph} {\Return $S$}
\If {$c > 0$} {
\ForEach {$w \in S \setminus \{v\}$} {
  Remove $w$ from the graph \\
  $S'$ $:=$ \findneigh($v$,$k$,$\delta$,$c-1$) \;
  Insert $w$ back into the graph \\
  \lIf{$S' \ne \{v\}$}{\Return $S'$}}
}
\Else {
\Return $\{v\}$
}
\end{algorithm}
The procedure uses brute force search to identify a $(k,\delta,c)$-isolated neighborhood $S$ of $v$ by guessing vertices in $N(S)$. In each execution, the procedure runs breadth-first search (BFS) to explore the vertices around $v$ until $k$ vertices are discovered or the BFS execution terminates. The latter happens when the size of the connected component containing $v$ is less than $k$.
Next the procedure checks if the set of visited vertices is a $(k,\delta,c)$-isolated neighborhood of $v$. If it is, the procedure outputs the set and terminates. Otherwise, if the set is not a $(k,\delta,c)$-isolated neighborhood, at least one of the vertices in the set must belong to $N(S)$. The procedure guesses a vertex $u$ (or in fact, enumerates over all possibilities), removes it from the graph and recursively searches for $S$ again. Since the neighborhood of $S$, if $S$ exists, has cut-size bounded by $c$, i.e., $|N(S)| \leq c$, the procedure has to recurse at most $c$ times to guess the entire set $|N(S)|$. This proves that the procedure discovers a $(k,\delta,c)$-isolated neighborhood of $v$, if it exists. If it does not then the procedure clearly returns $\{v\}$.

Since at each iteration there are at most $k$ possible guesses for a vertex in $N(S)$, the procedure is executed at most $O(k^c)$ times. Each execution can easily be implemented in $\poly(dk)$ time with at most $O(dk)$ queries in the BFS execution. Therefore, the total running time is $\poly(dk) \cdot k^c$ and the total query complexity is bounded by $O(dk^{c+1})$.
\end{proof}

\subsection{Finding Isolated Neighborhoods Covering a Vertex}

Let $v$ and $u$ be vertices in a graph $G=(V,E)$. When the values of $k, \delta, c$ are clear from context, we say that $u$ \emph{covers} $v$ if the isolated neighborhood found by \findneigh($u$,$k$,$\delta$,$c$) contains $v$. The following lemma states that one can efficiently find all vertices that cover a given vertex.

\begin{lemma}\label{lem:find-covering}
Let $G = (V,E)$ be a graph with maximum degree bounded by $d$. There is an algorithm that given a vertex $v \in V$, integers $k,c \geq 1$ and $\delta \in (0,1)$, finds all $u$ that cover $v$
in $\poly(cdk) \cdot k^{2c}$ time and with $O(dk^{2(c+1)})$ queries.
\end{lemma}

\begin{proof}
We modify \findneigh($v$,$k$,$\delta$,$c$) so that it does not terminate after finding an isolated neighborhood. Instead it just outputs it, and continues. This way it eventually lists all isolated $(u,k,\delta)$-neighborhoods containing $v$. Therefore, if $u$ covers $v$, then $u$ must be visited in the execution of the modified \findneigh. This implies that there are at most $O(k^{c+1})$ vertices $u$ that belong to an isolated neighborhood common with $v$,
and can therefore cover $v$. We can discover them in $\poly(dk) \cdot k^{c}$ time with $O(dk^{c+1})$ queries. Next we run \findneigh($v$,$k$,$\delta$,$c$) for each of them to find these that in fact cover $v$. We also remove the repetitions from the list. The total running time is $\poly(cdk) \cdot k^{2c}$ and the total number of queries is $O(dk^{2(c+1)})$.
\end{proof}

\subsection{Small Cover of Isolated Neighborhoods}

Recall that a cover for a family $\mathcal A$ of subsets of vertices is any subset $\mathcal B \subseteq \mathcal A$ such that $\bigcup_{T \in \mathcal A} T = \bigcup_{T \in \mathcal B} T$. We now show that a family of isolated neighborhoods of a vertex has a relatively small cover.

\begin{lemma}\label{lem:small_cover}
Let $\mathcal A$ be a family of isolated neighborhoods of a vertex $v$ in a graph $G=(V,E)$. Let the cut-size of all the neighborhoods be bounded by an integer $c$. There is a cover $\mathcal B \subseteq \mathcal A$ of size at most $(c+1)!$.
\end{lemma}

\begin{proof}
We assume that $\mathcal A$ is non-empty. Otherwise, the lemma holds trivially. We prove the lemma by induction on $c$. For $c=0$, any isolated neighborhood of $v$ is the connected component containing $v$. Therefore, $\mathcal A$ consists of one set, $(c+1)!=1$, and the lemma holds.

For the inductive step, assume that $c > 0$ and that the lemma holds for cut-size bounds lower than $c$. Without loss of generality, there is a vertex $u\in V$ that belongs to exactly one set $S \in \mathcal A$. Otherwise, we can keep removing arbitrary sets from $\mathcal A$, one by one, until this becomes the case, because every cover for the pruned family of neighborhoods is still a cover for the original family.

For all $T \in \mathcal A$, let $G_T$ be the subgraph of $G$ induced by the vertices in $T \setminus S$. For each $w \in N(S)$, we construct a family $\mathcal A_w$ of neighborhoods of $w$ as follows. We consider each $T \in \mathcal A$. If $w \in T$, we add to $\mathcal A_w$ the set $X$ of vertices in the connected component of $G_T$ that contains $w$. In this case, we say that $T$ was the \emph{progenitor} of $X$.

Let $G'$ be the subgraph of $G$ induced by $V \setminus S$. We claim that for each $w \in N(S)$ and each $Y \in \mathcal A_w$, the cut-size of $Y$ in $G'$ is bounded by $c - 1$. Let $Z \in \mathcal A$ be the progenitor of $Y$. Every vertex that belongs to the neighbor set of $Y$ in $G'$ also belongs to the neighbor set of $Z$ in $G$. Moreover, since $w \in Z$, $Z \ne S$, and does not contain $u$. Therefore, there is a vertex in $S$ that belongs to $N(Z)$ in $G$. This vertex does not appear in $G'$, which implies it does not belong to $N(Y)$ in $G'$. Due to our earlier observation that 
$N(Y)$ in $G'$ is a subset of $N(Z)$ in $G$, $|N(Y)| \le c - 1$.

We construct $\mathcal B$ as follows. We start with an empty set and insert $S$ into $\mathcal B$. By the inductive hypothesis, for every $w \in N(S)$, there is a cover $\mathcal B_w$ of size at most $c!$ for $\mathcal A_w$. For each neighborhood $Z$ in each of these covers, we add to $\mathcal B$, the progenitor of $Z$. This finishes the construction of $\mathcal B$. The size of $\mathcal B$ is bounded by $|N(S)| \cdot c! + 1 \le c \cdot c! + 1 \le (c+1)!$.

It remains to show that $\mathcal B$ is a cover of $\mathcal A$. Consider any vertex $t \in \bigcup_{Z \in \mathcal A}Z$. If $t$ belongs to $S$, we are done. Otherwise, let $Z_1\in\mathcal A$ be such that $t \in Z_1$. There is a path in $G$ that goes from $v$ to some $w \in N(S)$ via vertices in $S$, and then using only vertices in $Z_1 \setminus S$ it goes from $w$ to $t$. Let $G_\star$ be the subgraph of $G$ induced by $Z_1 \setminus S$. Let $Y_1$ be the set of vertices in the connected component of $G_\star$ that contains $w$. Due to the path from $w$ to $t$, $t \in Y_1$. By definition, $Y_1 \in \mathcal A_w$, so $t \in \bigcup_{Y \in \mathcal A_w} Y= \bigcup_{Y \in \mathcal B_w} Y$.
Let $Y_2 \in \mathcal B_w$ be any set containing $t$. Its progenitor $Z_2 \supseteq Y_2$ belongs to $\mathcal B$, and therefore, $t \in \bigcup_{Z \in \mathcal B} Z$, which finishes the proof.
\end{proof}

\begin{corollary}\label{cor:how_many_can_cover}
Let $v$ be a vertex in a graph. Let $k,c \ge 1$ be integers and let $\delta \in (0,1)$.
The number of vertices $u$, for which $u \in \mbox{\findneigh}(u,k,\delta,c)$ is bounded by $k\cdot(c+1)!$.
\end{corollary}

\section{The Partitioning Oracle (Proof of Theorem \ref{theorem:partition-treewidth})}

Here we prove the main claim of the paper. We show a global partitioning algorithm that can easily and efficiently be simulated locally. The global algorithm is likely to find the desired partition of the input graph.

\begin{algorithm}[t]
\caption{{\tt Global-Partitioning}($k$,$\delta$,$h$)}
\label{alg:global-treewidth}
\lForAll {$v \in V$} {set $v$ as not marked} \\
\ForEach {$v \in V$} {
$S_v := {}$\findneigh($v$,$k$,$\delta$,$2(h + 1)$) \\
$r_v := \mbox{uniformly random value in $(0,1)$}$
}

\ForEach {\rm $v \in V$ in increasing order of $r_v$} {
  $U := \{ w \in S_v : w$ is not marked$\}$\\
  \lForAll{$w \in U$}{$f[w] := U$ }  \\
  Mark all vertices in $U$
}

Output $f$
\end{algorithm}

\begin{algorithm}[t]
\caption{{\tt Local-Partitioning}($q$,$k$,$\delta$,$h$) for a vertex $q$}
\label{alg:local-treewidth}
$Q_q := $ the set of vertices that cover $q$ (see Lemma~\ref{lem:find-covering}) \\
Let $u$ be the vertex in $Q_q$ with the lowest $r_u$\\
$S_u := {}$\findneigh($u$,$k$,$\delta$,$2(h + 1)$) \\
$P := \emptyset$ \\
\ForEach {$w \in S_u$} {
$Q_w := $ the set of vertices that cover $w$ (see Lemma~\ref{lem:find-covering})\\
Let $u_w$ be the vertex in $Q_w$ with the lowest $r_{u_w}$\\
\lIf {$u = u_w$}{$P := P \cup \{ w \}$ }
}
\Return $P$
\end{algorithm}

\begin{proof}
We want to set our parameters $\delta$ and $k$ so that the following inequalities hold:
\begin{eqnarray}
\delta & \le & \frac{\eps}{100 \cdot (2h+3)! \cdot (1 + \log k + \log (2h+3))},\label{eq:delta_bound}\\
k&\ge&\frac{\const\, d^5(h+1)^5}{\delta \eps^2}.\label{eq:k_bound}
\end{eqnarray}
By combining them, we obtain
$$k \ge \frac{\const00 \cdot d^5 \cdot (2h+3)! \cdot (1 + \log k + \log (2h+3)!)}{\eps^3}.$$
The inequality can be satisfied by values of $k$ that do not grow faster than
$$O\left(\frac{d^5 \cdot h^{O(h)} \cdot \log(d/\eps)}{\eps^3}\right).$$
Then we take the maximum $\delta$ that satisfies Equation~\ref{eq:delta_bound}.

Consider the global partitioning algorithm described as Algorithm~\ref{alg:global-treewidth} with parameters set as described above.
The algorithm constructs a partition of the vertices in the graph. It starts with all vertices in the graph unmarked. For each vertex $v \in V$, the algorithm tries to find a $(k,\delta,2(h+1))$-isolated neighborhood of $v$ in $G$. If one such neighborhood is found, $S_v$ is set to that neighborhood. Otherwise, if no such neighborhood exists, $S_v$ is set to $\{v\}$. Next the algorithm starts partitioning the graph. It considers the vertices in random order. For each vertex $v$ in that order, all the unmarked vertices in $S_v$ constitute a component in the partition and get marked.

Clearly, at the end of the execution of Algorithm~\ref{alg:global-treewidth}, all vertices must be marked. Therefore, $f(v)$ is well defined for all $v \in V$. Also, observe that when $f(v)$ is defined, we have $f(u) = f(v)$ for every $u \in f(v)$. Therefore, Claims 1 and 2 of the theorem hold for the function $f$ computed by Algorithm~\ref{alg:global-treewidth}.

Let us bound the probability that the number of edges between different parts is greater than $\eps |V|$. By Lemma \ref{lem:treewidthlocalcut} with the lemma's $\eps$ set to $\eps/d$, there exists a function $g:V \to \mathcal P(V)$ with the following properties:
\begin{enumerate}
\item For all $v \in V$, $v \in g(v)$.
\item For all $v \in V$, $|g(v)| \leq k$.
\item For all $v \in V$, $g(v)$ is connected.
\item Let $\mathcal B$ be the subset of $V$ such that $v \in \mathcal B$ if and only if $g(v)$ is a $(k,\delta,2(h+1))$-isolated neighborhood of $v$ in $G$. The size of $\mathcal{B}$ is at least $\left(1-\frac{\eps}{20d}\right)|V|$.
\end{enumerate}

Let us group the edges between different components in the partition given by $f$. We distinguish two kinds of edges: the edges incident to at least one vertex in $V\setminus \mathcal B$ and the edges with both endpoints in $\mathcal B$. Observe that the total number of the former edges is at most $\frac{\eps}{20}|V|$. It remains to bound the number of the latter edges that are cut. Consider a vertex $v \in \mathcal B$. Let $Q_v$ be the set of vertices that cover $v$ and let $m_v = |Q_v|$. Via Corollary~\ref{cor:how_many_can_cover}, $m_v \leq k \cdot (2h+3)!$. Let $q_1, q_2, \ldots, q_{m_v} \in Q_v$ be the sequence of vertices that cover $v$ in increasing order of $r$, i.e., $r_{q_1} \leq r_{q_2} \leq \ldots  \leq r_{q_{m_v}}$. For each $j \in \{1,2,\ldots,m_v\}$, let $S_v^{(j)} = \{S_{q_1},S_{q_2},\ldots, S_{q_j} \}$, where $S_{q_i}$ is the isolated neighborhood found by \findneigh{} starting from $q_i$. Note that, since $r$ is random, $S_v^{(j)}$ and $q_j$ are random variables for all $j \in \{1,2,\ldots, m_v \}$.

For vertices $u,v \in \mathcal B$, we say $u$ \emph{marks} $v$ if $v \in S_u$ and $v$ is not marked before $u$ is considered. Also, we say a vertex $u \in \mathcal B$ is \emph{marking} if $u$ marks some vertex in the graph. It is clear from the definition that, for any $j \in \{1,2,\ldots,m_v\}$, if $q_j$ is marking, then $S_{q_j} \not\subseteq \bigcup_{i=1}^{j-1}{S_{q_i}}$. This implies that if $q_j$ is marking, $S_{q_j}$ must be a member of every cover of $S_v^{(j)}$. By Lemma~\ref{lem:small_cover}, there exists a cover $B_j \subseteq S_v^{(j)}$ such that $|B_j| \leq (2h+3)!$. Thus, whatever the first $j$ sets $S^{(j)}_v$ are, the probability that $j$-th vertex $q_j$ is marking is bounded by $(2h+3)!/j$.
Let the number of marking vertices in $Q_v$ be $a_v$. We have
\begin{eqnarray*}
E[a_v] &=& \sum_{j=1}^{m_v}{\Pr[q_j \textrm{ is marking}]} \leq \sum_{j=1}^{m_v}{\frac{(2h+3)!}{j}} 
\le (2h+3)! \cdot(1 + \log m_v) \\
&\leq& (2h+3)! \cdot (1 + \log k + \log (2h+3)!).
\end{eqnarray*}
Let $M \subseteq \mathcal B$ be the set of marking vertices in the graph. It holds
$$ \sum\nolimits_{u \in M}{|S_u|} = \sum\nolimits_{u \in \mathcal B} {a_u}.$$
Thus,
$$ E\left[ \sum\nolimits_{u \in M}{|S_u|} \right] \leq (2h+3)! \cdot (1 + \log k + \log (2h+3)!) \cdot |\mathcal B|.$$
Note that, for each marking vertex $u$, the number of edges that have exactly one end in $S_u$ is at most $d |S_u| \delta$. This is also an upper bound for the number of edges going out of $u$'s component in the partition. Therefore, the expected number of cut edges with both ends in $\mathcal B$ is at most
$$ \delta d \cdot (2h+3)! \cdot (1 + \log k + \log (2h+3)!)\cdot |\mathcal B| \leq \frac{\eps}{100} |V|.
$$
Thus, by Markov inequality, the probability that the number of edges in the second set is greater than $\frac{\eps}{10}|V|$ is at most $1/10$. Therefore, the probability that the total number of edges in both sets is less than $\eps |V|$ is at least $9/10$, as required by Claim 3.

Finally, observe that the size of each partition is trivially bounded by $k$, which is required by Claim 4.

We now show how Algorithm~\ref{alg:global-treewidth} can be simulated locally. Consider the local Algorithm~\ref{alg:local-treewidth} that given a query $q \in V$, computes $f[q]$.  The local algorithm starts by computing the set of vertices that cover $q$. Then among the vertices in this set, it finds the vertex $u$ with the smallest value $r_u$. It is clear that $u$ is the vertex that marks $q$, and thus, $f[q]$ is a subset of $S_u$. Next the local algorithm considers each vertex in $S_u$ and checks whether that vertex is also marked by $u$. If it is, the vertex should also be included in $f[q]$. Clearly local Algorithm~\ref{alg:local-treewidth} computes exactly the same set $f[q]$ as the global Algorithm~\ref{alg:global-treewidth}, provided the selected random numbers are the same.

Let us bound the number of queries to the input graph that Algorithm~\ref{alg:local-treewidth} makes to answer a query about $q$:
\begin{itemize}
	\item Lemma~\ref{lem:find-covering} shows that finding $Q_q$, the set of vertices that cover $q$, requires $O(dk^{2h+4})$ queries to the input graph.
	\item By Lemma \ref{lem:find-isolated-neighborhood}, finding $S_u$ takes at most $O(dk^{2h+3})$ queries to the input graph.
	\item Finally, for each $w \in S_u$, checking whether $w$ is marked by $u$ also takes $O(dk^{4h+6})$ queries to the input graph. Since $|S_u| \le k$, it takes at most $O(dk^{4h+7})$ queries to find the subset of vertices in $S_u$ that are marked by $u$.
\end{itemize}
Summarizing, the oracle has to make at most $O(dk^{4h+7})$ queries to the input graph to answer any query about $f$. 

We assume that a dictionary operation requires $O(\log{n})$ time for a collection of size $n$. We use a dictionary to keep the previously generated $r_v$ in order to maintain consistency. The running time of the oracle to answer a query about $f$ is then at most $\tilde O(k^{4h+O(1)} \cdot \log{Q})$, as stated in Claim 5 of the theorem.

Finally, the partition is computed such that it is independent of queries to the oracle. It is only a function of coin tosses that correspond to Step 5 of Algorithm~\ref{alg:global-treewidth}\footnote{As a technical requirement to make the argument valid, we also assume that the loop in Step 6 of Procedure \findneigh{} considers vertices of the input graph in order independent of queries to the oracle.}.
\end{proof}

\bibliographystyle{alpha}
\bibliography{paper}

\end{document}